\documentclass[conference]{IEEEtran}
\IEEEoverridecommandlockouts
\usepackage{tikz}
\usetikzlibrary{calc}
\usepackage{caption}
\usepackage{subcaption}
\usepackage{amsthm}
\usepackage{cite}
\usepackage{amsmath,amssymb,amsfonts}
\usepackage{mathtools}
\usepackage{algorithmic}
\usepackage{graphicx}
\usepackage{textcomp}
\usepackage{xcolor}
\usepackage{bbold}
\usepackage{mathrsfs}

   \newtheorem{Corollary}{Corollary}
   \newtheorem{Proposition}{Proposition}
   \newtheorem{Definition}{Definition}
   
   \usepackage{titlesec}
  \titlespacing\subsection{0pt}{0pt plus 1pt minus 2pt}{0pt plus 1pt minus 2pt}
   \titlespacing\section{0pt}{5pt plus 1pt minus 2pt}{5pt plus 1pt minus 2pt}
\newcommand{\ac}[1]{\textcolor{black}{#1}}
\newcommand{\nh}[1]{\textcolor{black}{#1}}

\newcommand{\nn}[1]{\textcolor{black}{#1}}
\newcommand{\hs}[1]{\textcolor{black}{#1}}

\begin{document}

\title{Rate Outage and Meta-Distribution for Uplink Networks in Finite Block-Length Regime }
\author{
\IEEEauthorblockN{Nourhan Hesham, Hesham ElSawy, Jahangir Hossain, and Anas Chaaban}
\thanks{%

 N. Hesham, J. Hossain, and A. Chaaban are with the School of Engineering, University of British Columbia, Kelowna, BC V1V 1V7, Canada (e-mail: \{nourhan.soliman,anas.chaaban,jahangir.hossain\}@ubc.ca), and N. Hesham is on leave from the Faculty of Engineering, Cairo University, Cairo, Egypt.

H. ElSawy is with the School of Computing, Queen’s University, Kingston, ON K7L 2N8, Canada. (e-mail: hesham.elsawy@queensu.ca).

This publication is based upon work supported by King Abdullah University of Science and Technology (KAUST) under Award No. OSR-2018-CRG7-3734.                 

}
}

\maketitle
\begin{abstract}
    With the expected proliferation of delay constrained applications, future communication technologies are pushed towards using short codes. The performance using short codes cannot be inferred through classical channel capacity analysis, which intrinsically assumes long codes and vanishing frame error rate (FER). This paper studies the performance of an uplink large-scale network in the finite blocklength regime. Bounds on the spatially averaged rate outage probability as well as the coding rate meta distribution are derived. The results reveal the exact achievable rate for a given blocklength and FER, and demonstrate the discrepancy between the actual network rate and idealistic classical channel capacity.
\end{abstract}

\begin{IEEEkeywords}
Coding Rate Meta Distribution, Finite Blocklength Regime, Large-Scale Networks, Rate Outage Probability, Stochastic Geometry.
\end{IEEEkeywords}
\section{Introduction}
% Motivation
Future technologies are expected to enable services that require ultra-reliable low-latency communication. Achieving low latency requires the use of short codes. However, conventional characterizations of ergodic rate and outage are derived under the assumption of using infinitely long codewords and vanishing frame error rate (FER), which we call henceforth the asymptotic regime (AR). {To assess the performance of} short codes, Polyanskiy {\it et al}.~\cite{polyanski,BlockFadingPolyankiy} characterized the maximum achievable rates in the finite blocklength regime (FBR) for additive white Gaussian noise (AWGN) and block fading channels, respectively,  i.e., under a given codelength and FER. In this paper, we focus on studying the performance of large-scale networks in the FBR using stochastic geometry.

% Literature Survey
Most works that studied the performance of large-scale networks focused on the AR. For instance, in the uplink (UL) which is the focus of this work, \cite{Uplink1_Hesham} studied the rate outage probability and average spectral efficiency of a single and a multi-tier network using truncated channel inversion power control. In~\cite{UL_1}, the performance of \nh{dense heterogeneous cellular networks} in the AR under a fractional power control strategy is studied and the signal to interference-plus-noise-ratio (SINR) is characterized as a function of the association rules and power control parameters. {Another important performance metric studied in many works is the meta-distribution, which provides information about the \hs{percentile} of users experiencing a specific quality of service \hs{for an arbitrary, yet fixed, network deployment}.} In~\cite{Meta,Meta1}, the authors analyzed the signal-to-interference ratio (SIR) meta distribution for UL networks in the AR. In contrast, \cite{TWC_paper,Uplink4} considered the FBR. \cite{TWC_paper} studied the average coding rate, rate outage probability, and coding rate meta-distribution for a downlink (DL) large-scale network. However, \cite{Uplink4} studied the delay-bound violation probability under a given end-to-end latency bound in an UL multi-cell interference channel, and maximized the transmission reliability by optimizing the transmission rate.

To the best of the authors' knowledge, the performance of large-scale UL networks in the FBR in terms of rate outage probability and the coding rate meta-distribution has not been investigated. Note that the work in \cite{TWC_paper} studies a large-scale network in the FBR, {however} \cite{TWC_paper} focuses on the DL which is operationally different {from} the UL {that} uses power control. {Also, the UL network is not homogeneous as the DL network, which in return affects the analysis.} {In particular,} this paper studies {these metrics in an} UL large-scale networks in the FBR using stochastic geometry tools, while assuming fractional path-loss inversion power control~\cite{Meta}. {The} outage probability measures the fraction of failed transmissions at a given data rate, while the meta distribution provides fine-grained information about network performance in the form of the percentile of users that can achieve a specific minimum data rate at a given FER threshold.

\section{System Model}\label{sec:SystemModel}

 We consider a single-tier UL large-scale network utilizing universal frequency reuse {causing inter-cell interference,} combined with orthogonal multiple access (OMA) to avoid intra-cell interference.\footnote{{OMA schemes avoid intra-cell interference through an orthogonal allocation of resources.}} Hence, only one user per cell has access to the same resource block. \nh{However, its received signal at the BS is interfered by a large number of users from other cells.}
 The base stations (BSs) constitute a 2-dimensional (2-D) spatial {Poisson point process} (PPP) $\Psi$ with intensity $\lambda $. \nn{The users' equipment (UEs) point process $\Phi=\{u_i; i \in \mathbb{N}\}$ is modeled by associating each UE with its geographically closest BS.} \ac{The intended UE sends information to its serving BS using codewords from a code with rate $R$, a target FER $\bar{\epsilon}$, and length $n$ symbols. The value of $n$ can reflect latency and/or energy consumption constraints.\footnote{\nh{In a network, one of the dominant delays is the transmission delay, especially in dense networks. The transmission delay is a function of the blocklength, hence it should be designed to satisfy the latency constraint.} \nn{In applications with stringent energy consumption requirements, the average energy of a block should satisfy energy consumption constraints.}} Besides, \nh{the target FER $\bar{\epsilon}$ can be defined following a reliability constraint,} \nn{so that} the FER $\epsilon$ of the network conditioned on the SINR should be less than or equal $ \bar{\epsilon}$.} Treating interference as noise and following Polyanskiy {\it et al.} \cite{polyanski}, the coding rate at SINR $\Omega$, blocklength $n$, and FER $\epsilon$ can be approximated by
\begin{align}
    R_{n,\epsilon}(\Omega)= \log_2(1+\Omega)-\frac{\sqrt{V(\Omega)}Q^{-1}(\epsilon)}{\sqrt{n}}+\frac{1}{2 n}\log_2(n),\nonumber
\end{align}
where $V(\Omega)=\frac{\Omega(\Omega+2)\log_2^2(e)}{(\Omega+1)^2}$ is the channel dispersion.
In practical networks, the SINR is unknown at the transmitter and hence, the transmitter encodes at a constant (target) rate $R_t$ and blocklength $n$. As a result, the resulting FER $\epsilon$ conditioned on SINR $\Omega$ may or may not satisfy the desired design FER $\bar{\epsilon}$ due to stochastic fading and aggregate interference. In this case, it is important to keep this conditional FER below a threshold to avoid having excessive frame errors.   

\ac{For this network, the BS receives a signal $y$ from its connected UE that lies at a distance $r_0$. The distance $r_0$ follows a Rayleigh distribution~\cite{r_o}, i.e., it is distributed as $ f_{r_0}(r_0)=   2 \pi \lambda r_0 \exp\{-\pi \lambda r_0^2\},\ \  0<r_0<\infty$, due to the nearest BS association assumption. The BS also receives interference from other UEs located at distances $u_1,\ u_2,...\ $, given by}
\begin{align}
    y= \sqrt{P_0} {h}_0 r_0^\frac{-\eta}{2}s_0+\sum\limits_{u_i \in \Phi\setminus \{r_0\}} x_i \sqrt{P_i} h_i u_i^\frac{-\eta}{2} s_i +w,
\end{align}
where {$P_0$ (resp. $P_i$) is the transmit power of the intended UE (resp. UE$_i$), $h_0$ (resp. $h_i$) is the channel coefficient from intended UE (resp. UE$_i$) to the BS, $\eta$ is the path loss exponent, $s_0$ (resp. $s_i$) is a unit average power codeword symbol transmitted by the intended UE (resp. $UE_i$), $x_i\sim{\rm Bernoulli}(\delta)$ is an activity indicator for UE $i$, and $w$ is the AWGN noise with noise power $N_0$. {Utilizing} fractional path-loss inversion, $P_i=\rho_o r_i^{\alpha\eta}$} where \ac{$\alpha$ is a compensation factor which determines the fraction of channel inversion (i.e., $\alpha=1$ refers to full channel inversion whereas $\alpha=0$ refers to no channel inversion), $\rho_o$ is the power control parameter ({at $\alpha=1$, $\rho_o$ is the average receive power, and at $\alpha=0$, $\rho_o$ is the transmit power.}), and $r_i$ is the distance between the UE and its serving BS.} A block fading channel model is assumed to follow a Rayleigh distribution ($h_0\sim \mathcal{CN}(0,1)$, $h_i$ and $h_0$ are i.i.d), where the channel coefficients remain constant for $L$ consecutive symbols and change to an independent realization in the next $L$ symbols. To ensure that the channel remains constant during the transmission of $n$ symbols, we require $n= L/l$ for some integer $l\geq 1$.\footnote{$l$ represents that number of frames that can be transmitted within one coherence interval. The analysis in this paper can be simulated by averaging over all coherence intervals and considering one codeword per coherence interval.} The SINR is given by

\begin{align}
  \Omega=\frac{{P_0} |{h}_0|^2 r_0^{-\eta}}{\sum\limits_{u_i \in \Phi\setminus \{r_0\}} x_i {P_i} |h_i|^2 u_i^{-\eta} +N_0}=\frac{{\rho_o} |{h}_0|^2 r_0^{-\eta(1-\alpha)}}{\mathcal{I} +N_0},
\end{align}

\section{Analysis}\label{sec:Outage_UL}

The rate outage probability and the coding rate meta distribution are characterized in this section, providing a characterization of the network performance.

The serving BS wants to decode the information sent by the intended UE indexed by $0$, which is distorted by noise $N_0$ and {aggregate} interference $\mathcal{I}$ received from {all} {active} UEs using the same UL channel frequency in other cells. The Laplace transform of $\mathcal{I}$ is given by
\begin{align}\label{eq:LT}
  \mathcal{L}_{\mathcal{I}}(s)\hspace{-0.07cm}=\hspace{-0.05cm}\exp\hspace{-0.07cm}\left\{\hspace{-0cm} \int\limits_{\frac{r_0^{\eta(1-\alpha)}}{\rho_o}}^{\infty}\hspace{-0.1cm} \frac{-2\gamma \left( \hspace{-0.07cm} 1+\alpha,\pi \lambda (\rho_o w)^{\frac{2}{\eta(1-\alpha)}}\hspace{-0.07cm} \right)\delta s}{\rho_o^{\frac{-2}{\eta}}(\pi \lambda)^{\alpha-1}\eta w^{1-\frac{2}{\eta}} (w+s)}  dw\hspace{-0.07cm}\right\},
\end{align}
where $\gamma(a,b)\hspace{-0.1cm}=\hspace{-0.1cm}\int_{0}^{b} t^{a-1} e^{-t}\ dt$. See Appendix A for the derivation of $\mathcal{L}_{\mathcal{I}}(s)$.

\subsection{Rate Outage Probability} 
The definition of the rate outage is provided as follows.
\begin{Definition}\label{def:outage}
Given a target rate $R_t$, the outage probability of large-scale UL network for a blocklength $n$ and {target FER} $\bar{\epsilon}$, is defined as $\mathcal{O}_{n,\bar{\epsilon}}(R_t)=\mathbb{P}(R_{n,\bar{\epsilon}}(\Omega)<R_t)$.
\end{Definition}

{Due to the added channel dispersion term $V(\Omega)$, the rate outage probability is difficult to evaluate. By studying the behavior of the channel dispersion as a function of SINR $\Omega$, we have found that it increases from $0$ to $\log_2^2(e)$ as the SINR increases from $0$ to $\infty$. Hence, the outage probability can be bounded as follows.}
%\begin{align}
%\mathcal{O}_{n,\epsilon}^{\text{\tiny DL}}(r_0,R_t)\approx 1\hspace{-.05cm}-e^{\hspace{-.1cm}-\frac{2^{R_t+a}-1 }{\alpha r_0^{-\eta}}}\mathcal{L}_{\mathcal{B}}\left(\frac{2^{R_t+a}-1 }{\mathcal{P} r_0^{-\eta}}\right)
%\end{align}
%where $a= \sqrt{ \frac{\log_2^2(e)}{ n} } Q^{-1}(\epsilon)$, and $\mathcal{P}$ is the downlink transmit power. Using the same analogy, the rate outage probability of the UL large-scale network is proposed in the following Corollary.
%The rate outage probability defined in Def. \ref{def:outage}, under fractional channel inversion power control with a compensation factor $\alpha$, average received power threshold $\rho_o$, and noise power $N_0$, is upper bounded by

\begin{Proposition}\label{Proposition:1}
The rate outage probability \hs{under the considered system model} is upper bounded by
\begin{align}\label{eq:Outage_UB}
  & \bar{ \mathcal{O}}_{n,\bar{\epsilon}}(R_t,\Upsilon)\hspace{-0.05cm}= \mathbb{F}_{\Omega}(\max(\Upsilon,2^{R_t+a_{n,\bar{\epsilon}}(\infty)-b_n}-1))\nonumber\\
   &\ \ +\mathbb{F}_{\Omega}(\min(\Upsilon,2^{R_t+a_{n,\bar{\epsilon}}(\Upsilon)-b_n}-1))-\mathbb{F}_{\Omega}(\Upsilon)
\end{align}
for any $\Upsilon>0$, and lower bounded by 
\begin{align}\label{eq:Outage_LB}
    \underline{\mathcal{O}}_{n,\bar{\epsilon}}(R_t,\Lambda)=&\mathbb{F}_{\Omega}(\min(\Lambda,2^{R_t-b_n}-1))-\mathbb{F}_{\Omega}(\Lambda)\nonumber\\
    &+\mathbb{F}_{\Omega}(\max(\Lambda,2^{R_t+a_{n,\bar{\epsilon}}(\Lambda)-b_n}-1))
\end{align}
for a $\Lambda>0$, where $a_{n,\bar{\epsilon}}(x)= \sqrt{ \frac{V(x)}{ n} } Q^{-1}(\bar{\epsilon})$, $b_n=\frac{1}{2n}\log_2(n)$, $\mathcal{L}_{\mathcal{I}}(\cdot)$ is provided in \eqref{eq:LT}, and $\mathbb{F}_\Omega(x)$ is the cumulative probability function (CDF) of the $\rm SINR$ $\Omega$ given by
\begin{align}
   \hspace{-0.15cm} \mathbb{F}_{\Omega}(\Omega)=1\hspace{-.07cm}-\hspace{-.15cm}\int_{0}^{\infty}\hspace{-.3cm} 2\pi\lambda r_0 &e^{\hspace{-.1cm}-\pi\lambda r_0^2-\frac{\Omega r_0^{\eta(1-\alpha)} }{\rho_o /N_0}}\hspace{-.12cm}\mathcal{L}_{\mathcal{I}}\hspace{-.07cm}\left(\frac{\Omega \rho_o^{-1}}{r_0^{-\eta(1-\alpha)}}\hspace{-.1cm}\right) d r_0,
\end{align}

\end{Proposition}
\begin{proof}
The upper bound is derived in Appendix \ref{Appendix:UpperBound} and the lower bound is derived in Appendix \ref{Appendix:LowerBound}.
\end{proof}
\ac{In Proposition 1, the upper and lower bounds \nh{hold for arbitrary values} of $\Upsilon$ and $\Lambda$, respectively. However, to achieve a tight bound, \nh{these parameters} have to be optimized and this can be done numerically.}

%In practice, networks do operate at moderate to high SINR, therefore this upper bound precisely characterizes the performance of the network at a large range of moderate to high $R_t$, as shown in Sec. \ref{sec:Results}. Also, the upper bound proposed provides a guaranteed performance at all values of $R_t$. 

To reduce the complexity of the expressions, simplified upper and lower bounds are provided \hs{in the following Corollary}.

\begin{Corollary}\label{Corollary:1}
    A simplified upper bound can be obtained by setting $\Upsilon=\infty$ leading to 
    \begin{align}
         \bar{ \mathcal{O}}_{n,\bar{\epsilon}}(R_t,\infty)\hspace{-0.05cm}&= \mathbb{F}_{\Omega}(2^{R_t+a_{n,\bar{\epsilon}}(\infty)-b_n}-1),
    \end{align}
    and a simplified lower bound can be obtained by setting $\Lambda=0$ leading to
        \begin{align}
         \underline{ \mathcal{O}}_{n,\bar{\epsilon}}(R_t,0)\hspace{-0.05cm}&= \mathbb{F}_{\Omega}(2^{R_t-b_n}-1),
    \end{align}
    
\end{Corollary}
The simplified upper bound is fairly tight at large $R_t$ as we shall see in Fig. \ref{fig:Rate_Outage_n}, whereas the simplified lower bound is approximately the rate outage probability in the AR. 
%\textcolor{red}{This upper bound $\bar{ \mathcal{O}}_{n,\bar{\epsilon}}(R_t)$ is generated by upper bounding the term $\sqrt{1-\frac{1}{(1+\Omega)^2}}$ by $1$. By studying this term, we would find that it ranges between $[0,1]$. For moderate and high SINR $\Omega$, this term tends to $1$ and the upper bound is shown to be a tight approximation. For low SINR, this term tends to $0$ and the outage probability in the AR could be used. However, in practice, networks do operate at moderate to high SINR, therefore this upper bound precisely characterizes the performance of the network, as shown in Sec. \ref{sec:Results}.}

%Next, numerical simulations are presented and a detailed discussion is conducted to analyze the performance of the two techniques and compare between them in a FBR.

\subsection{{Coding Rate Meta Distribution}}\label{sec:MetaDistribution}

 The meta distribution provides more detailed information about the quality of service in terms of coding rate by providing the fraction of users achieving a minimum data rate of $R_t$ {with probability at least $p_{th}$} at a blocklength $n$ and a target FER $\bar{\epsilon}$. Denote by $\mathbb{P}_s(R_t,n,\bar{\epsilon})$ the conditional success probability \nh{for an arbitrary, yet fixed, network realization $\Phi\cup\Psi$}, {defined as the achievability} of a rate $> R_t$ at a blocklength $n$ and FER  $\bar{\epsilon}$. Then, $\mathbb{P}_s(R_t,n,\bar{\epsilon})$ is given by
\begin{align}\label{eq:Ps_exact}
&\mathbb{P}_s(R_t,n,\bar{\epsilon})=\mathbb{P}^{^^21}(R_{n,\bar{\epsilon}}(\Omega) > R_t|\Phi,\Psi)\nonumber\\
    &\hspace{-0.1cm}=\mathbb{P}^{!}\left(\log_2(1+\Omega)-\sqrt{\frac{V(\Omega)}{n}} Q^{-1}(\bar{\epsilon})+b_n> R_t|\Phi,\Psi\right),
\end{align}
where $\mathbb{P}^{^^21}(\cdot)$ is the reduced Palm probability \cite[Def. 8.8]{reducedPalm_book}.

Note that $\mathbb{P}_s(R_t,n,\bar{\epsilon})$ is a random variable which depends on the relative locations between the UEs and the BSs within the network. Its moments are defined as
\begin{align}\label{eq:moments_exact}
    M_b=\mathbb{E}\left(\mathbb{P}_s(R_t,n,\bar{\epsilon})^b\right),
\end{align}
and can be used to evaluate the meta distribution via Gil-Pelaez theorem as follows~\cite{Meta1}
\begin{align}\label{eq:meta_exact}
    F_{R_t}(p_{th},n,\bar{\epsilon})&=\mathbb{P}(\mathbb{P}_s(R_t,n,\bar{\epsilon})>p_{th})\\
      &=\frac{1}{2}+\frac{1}{\pi}\int_{0}^{\infty}\frac{\mathcal{I}m(e^{-t \log(p_{th})} {M}_{jt})}{t} dt,\nonumber
\end{align}
{where $\mathcal{I}m(\cdot)$ denotes the imaginary part and $M_{jt}$ is obtained by replacing $b$ by $j t$ in \eqref{eq:moments_exact}, where $j$ is the imaginary unit.}

\ac{The expressions provided in \eqref{eq:Ps_exact}, \eqref{eq:moments_exact}, and \eqref{eq:meta_exact} can be evaluated {numerically}. To obtain an analytic expression, we utilize the moment matching method with \nn{a Beta distribution, proposed in \cite{Meta1}, to approximate \eqref{eq:meta_exact} to obtain}} 
%a simple yet tight \ac{approximation of the} distribution \ac{of the probability of success was} proposed \ac{for the AR} in~\cite{Meta1} using the beta distribution given by
\vspace{-0.15cm}
\begin{align}\label{eq:beta}
    f_X(x)=\frac{x^{\frac{\kappa(\zeta+1)-1}{1-\kappa}}(1-x)^{\zeta-1}}{\text{B}\left(\frac{\kappa \zeta}{1-\kappa},\zeta\right)},
\end{align}
where $\text{B}(\cdot,\cdot)$ is the Beta function, $\kappa={M}_1$ and the $\zeta=\frac{(M_1-M_2)(1-M_1)}{M_2-M_1^2}$. Hence, the meta distribution can be computed as the \ac{complementary cumulative distribution function (CCDF)} of the Beta distribution. However, $M_1$ and $M_2$ provided in \eqref{eq:moments_exact} are difficult to evaluate due to the added channel dispersion term $V(\Omega)$. Instead, a lower bound on the conditional success probability and an approximation of the meta-distribution are given in the following proposition.

%\begin{figure}
 %   \includegraphics[width=0.9\linewidth,height=6cm,trim={0.5cm 0.3cm 1cm 0.7cm},clip]{Outage_vs_Rt_new4.eps}
  %  \caption{The rate outage probability versus $(R_t)$ for $\rho_o\hspace{-0.07cm}=\hspace{-0.07cm}-60\ \text{dBm}$, $\bar{\epsilon}\hspace{-0.07cm}\in\hspace{-0.07cm}\{10^{-2},10^{-6}\}$, $\sigma_w^2\hspace{-0.07cm}=\hspace{-0.07cm}-90\ \text{dBm}$, and $n\in\{128,2048\}$ }
   % \label{fig:Rate_Outage_n}
%\end{figure}
\begin{figure*}[t]
     \begin{subfigure}[b]{0.48\textwidth}
     
         \centering
            \includegraphics[height=0.8\linewidth,width=1\linewidth,trim={0cm 0cm 1cm 0.7cm},clip]{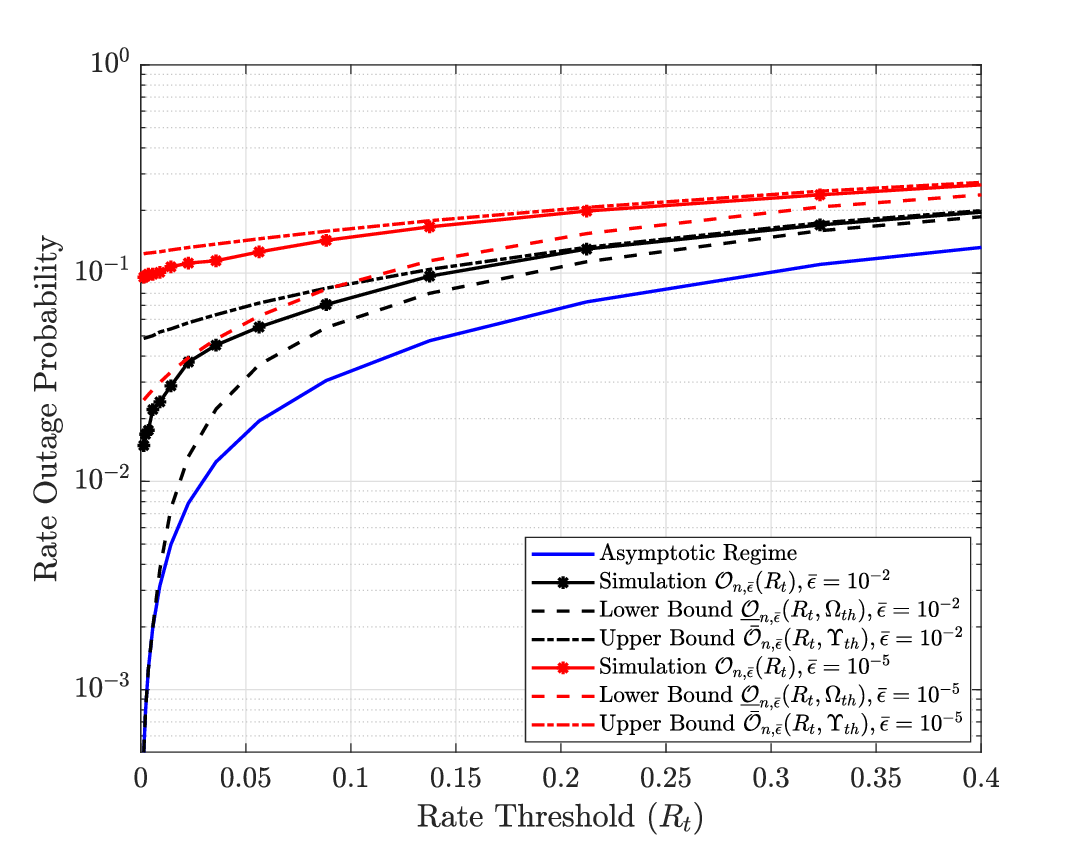}
         \caption{$n=128$ }
         \label{fig:Outage_128}
     \end{subfigure}
     \begin{subfigure}[b]{0.48\textwidth}
        \centering        \includegraphics[height=0.8\linewidth,width=\linewidth,trim={0cm 0cm 1cm 0.7cm},clip]{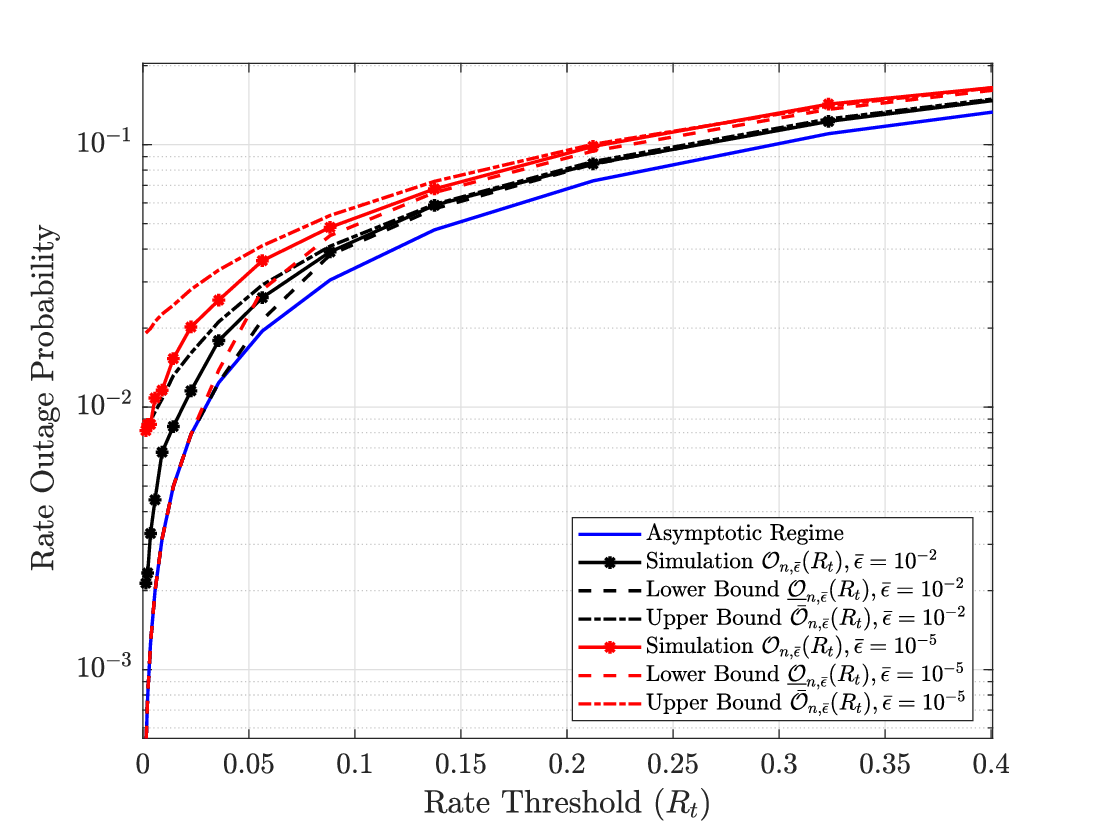}
         \caption{$n=2048$}
         \label{fig:Outage_2048}
     \end{subfigure}
     \caption{The rate outage probability versus $(R_t)$ for $\rho_o\hspace{-0.07cm}=\hspace{-0.07cm}-60\ \text{dBm}$, and $\sigma_w^2\hspace{-0.07cm}=\hspace{-0.07cm}-90\ \text{dBm}$ }
     \label{fig:Rate_Outage_n}
\end{figure*}

\begin{figure*}
\begin{align}\label{eq:approx_moments}
    &\hat{M}_b=\mathbb{E}\left(\mathbb{\hat{P}}_s(R_t,n,\bar{\epsilon})^b\right)\nonumber\\
     &=\hspace{-0.1cm}\int_0^\infty \hspace{-0.2cm} 2\pi\lambda r_0 \exp\left\{ -\pi\lambda r_0^2-\int_{r_0}^{\infty}\frac{2\gamma\left(1+\alpha,\pi\lambda(\rho_o w)^\frac{2}{\eta(1-\alpha)}  \right)}{(\pi\lambda)^{\alpha-1} \rho_o^\frac{-2}{\eta}\eta w^{1-\frac{2}{\eta}}}   \left( 1-\left( \frac{p}{1+\frac{\theta)r_0^{\eta(1-\alpha)}}{\rho_o w }}+(1-p)\right)^b \right) e^{\frac{-\theta b N_0}{\rho_o r_0^{-\eta(1-\alpha)}}} d w \right\},
\end{align}
\hrulefill
\end{figure*}
\begin{Proposition}
The conditional success probability can be lower bounded by 
\vspace{-0.3cm}
\begin{align}\label{eq:lower_bound}
    \mathbb{\hat{P}}_s(R_t,n,\bar{\epsilon})=\hspace{-0.5cm}\prod_{r_i\in \zeta\setminus \{r_0\}} \hspace{-0.15cm} \left(\hspace{-0.1cm} \frac{p}{1+\frac{\theta P_i r_i^{-\eta}}{\rho_o r_0^{-\eta(1-\alpha)}}}\hspace{-0.07cm}+\hspace{-0.07cm}(1-p)\hspace{-0.1cm}\right)  e^{\frac{-\theta N_0}{\rho_o r_0^{-\eta(1-\alpha)}}},
\end{align}
where $\theta=2^{R_t+a_{n,\bar{\epsilon}}(\infty)-b_n}-1$. 
Using \eqref{eq:lower_bound} and the Beta distribution approximation, the meta distribution \hs{in \eqref{eq:meta_exact}} can be approximated as 
\begin{align}\label{eq:meta_beta}
      &\hat{F}_{R_t}(p_{th},n,\bar{\epsilon})\hspace{-0.07cm}\approx \hspace{-0.07cm}\nonumber\\&\hspace{0.5cm}1\hspace{-0.07cm}-\hspace{-0.07cm} I_{p_{th}}\hspace{-0.1cm}\left(\frac{\hat{M}_1(\hat{M}_1-\hat{M}_2)}{(\hat{M}_2-\hat{M}_1^2)},\frac{(1-\hat{M}_1)(\hat{M}_1-\hat{M}_2)}{(\hat{M}_2-\hat{M}_1^2)}  \right),
\end{align}
where $I_{z}(x,y)=\frac{1}{\text{B}(x,y)}\int_{0}^{z} t^{x-1} (1-t)^{y-1}dt$ is the regularized incomplete Beta function, and the $\hat{M}_1$ and $\hat{M}_2$ are the moments of $\mathbb{\hat{P}}_s(R_t,n,\bar{\epsilon})$ provided in \eqref{eq:approx_moments} at the top of next page.
\end{Proposition}

\begin{proof}
Starting from \eqref{eq:Ps_exact}, we upper bound the $\sqrt{V(\Omega)}$ by $\log_2(e)$ which is considered a simple yet convenient bound (it quickly approaches $\log_2(e)$ as $\Omega$ increases), leading to $\mathbb{P}_s(R_t,n,\bar{\epsilon})\geq \mathbb{\hat{P}}_s(R_t,n,\bar{\epsilon})$ where
\begin{align}
    \mathbb{\hat{P}}_s(R_t,n,\bar{\epsilon})&=\mathbb{P}^{!}\left(\log_2 (1+\Omega)- a_{n,\bar{\epsilon}}(\infty)+b_n > R_t|\zeta\right).\nonumber
\end{align}
Consequently, 
\begin{align}
    &\mathbb{\hat{P}}_s(R_t,n,\bar{\epsilon})=\mathbb{P}^{!}\left(\Omega > \theta|\zeta\right) \nonumber\\
    &=\mathbb{E}_{h_i,x_i}\left\{\prod_{r_i\in \zeta\setminus \{r_0\}} e^{\frac{-\theta x_i |h_i|^2 P_i r_i^{-\eta}}{\rho_o r_0^{-\eta(1-\alpha)}}}  \right\}e^{\frac{-\theta N_0}{\rho_o r_0^{-\eta(1-\alpha)}}},\nonumber
\end{align}
By averaging over the channel gains, we obtain \eqref{eq:lower_bound}.
The moments of $\mathbb{\hat{P}}_s(R_t,n,\bar{\epsilon})$ are given in~\eqref{eq:approx_moments}. Substituting these moments in the Beta distribution \eqref{eq:beta} yields \eqref{eq:meta_beta}.
\end{proof}

\section{Numerical Results}\label{sec:Results}

We consider a network with intensity $\lambda=2\  \text{BSs}/\text{km}^2$, power control parameter $\rho_o=\{ -60\}\ \text{dBm}$, and path loss exponent $\eta=4$. All the results are validated using Monte Carlo simulations. In the simulation scenario, we set one PPP realization for the BSs in an area of $ 30 \times 30\ \rm{km}^2$ area. One UE is placed \hs{uniformly within the coverage of each BS.}

For the outage probability, we vary the network and the fading realization {and calculate the instantaneous SINR $\Omega$ and the instantaneous rate $R_{n,\bar{\epsilon}}(\Omega)$. Then, we compare the instantaneous rate to the rate threshold $R_t$ and average to evaluate the outage probability.} Further, we numerically optimize the bounds at each $R_t$ to optimally choose $\Lambda$ and $\Upsilon$. In Fig.~\ref{fig:Rate_Outage_n}, the rate outage probability is plotted versus $R_t$ for $n=\{128,2048\}$ and target FER $\bar{\epsilon}\hspace{-0.1cm}\in\hspace{-0.1cm}\{10^{-2},10^{-5}\}$. As shown in the figure, the rate outage upper and lower bounds in \eqref{eq:Outage_UB} and \eqref{eq:Outage_LB} are tight bounds on the performance of the network especially at moderate to high $R_t$. It is worth noting that the proposed upper bound provides a guaranteed performance at all values of $R_t$ and can be used as an approximation for moderate to high $R_t$. \ac{While it is evident that a gap exists between the rate outage probabilities in the AR and FBR, quantifying and investigating the behavior of this gap is essential. {This gap is due to the term $\sqrt{\frac{V(\Omega}{n}}Q^{-1}(\epsilon)$, which is a function of SINR, which in turn depends on the network scale and density. When the SINR is high/low the $V(\Omega)$ increases/decreases, increasing/decreasing the gap between the FBR and the AR performance. However, the SINR is random and hence the gap needs to be characterized}.} This gap increases as the target FER $\Bar{\epsilon}$ becomes more stringent and the blocklength $n$ becomes shorter, which makes the rate outage probability in the AR for short blocklengths (e.g. $n=128$) misleading, which motivates the analysis in FBR. Moreover, it is noticed that as the blocklength increases the rate outage decreases approaching the performance in the AR.

    \ac{To investigate the performance of the simplified bound provided in Corollary \ref{Corollary:1}, the simplified upper bound versus $R_t$ is plotted in Fig. \ref{fig:Simplified_Bound} with the optimized upper bound and the exact performance of the network. The simplified upper bound provides a simple expression that can be easily evaluated but it is not tight for low $R_t$. However, for moderate $R_t$, it is shown to be tight and can be used as an approximation. The tightness of this bound at moderate to large $R_t$ and the existence of a gap between the exact outage probability and the bound is explained next.}

    \ac{In Fig. 1 and 2, it is observed that a gap exists between the optimized bounds and the exact outage probability of the network. This gap increases at low $R_t$ and diminishes at moderate to high $R_t$. Hence, an analysis is conducted to gain insights into the behavior of this gap and to investigate the significance of the parameters $\Upsilon$ and $\Lambda$, and their impact on the outage probability. We {start} with the outage definition, where $\mathcal{O}_{n,\bar{\epsilon}}(R_t)=\mathbb{P}(R_{n,\bar{\epsilon}}(\Omega)<R_t)=\mathbb{P}(C(\Omega)<R_t+a_{n,\bar{\epsilon}}(\Omega)-b_n)$. The outage probability is calculated by {integrating} the {the pdf of the SINR $\Omega$ such that $C(\Omega)<R_t+a_{n,\bar{\epsilon}}(\Omega)-b_n$}. Hence, we plotted the capacity $C(\Omega)$ and the term $R_t+a_{n,\bar{\epsilon}}(\Omega)-b_n$ as functions of SINR $\Omega$ in Fig. 3. The term $R_t+a_{n,\bar{\epsilon}}(\Omega)-b_n$ {increases} from $R_t-b_n$ to $R_t+a_{n,\bar{\epsilon}}(\infty)-b_n$ as SINR increases. To simplify our analysis, we {upper bound} this function by a step function that {changes value} from $R_t+a_{n,\bar{\epsilon}}(\Upsilon)-b_n$ to $R_t+a_{n,\bar{\epsilon}}(\infty)-b_n$ at $\Upsilon$, {which} {is also plotted {in Fig. 3 for} optimized values of $\Upsilon$.} {It is worth noting that the optimal $\Upsilon$ is chosen numerically such that $\bar{ \mathcal{O}}_{n,\bar{\epsilon}}(R_t,\infty)$ is minimized.} For small $R_t$, we find that the {approximating} step function overestimates the region at the optimal choice of $\Upsilon$, leading to an integration over undesired region {thus increasing the gap between the bound and} the actual outage probability. For moderate to large $R_t$, the point of intersection of the approximated step function and the capacity approximately {matches} with the point of intersection of the two exact curves {for any choice of $\Upsilon$}. Hence, it does not overestimate the {integration interval}. The significance of $\Upsilon$ is evident at low $R_t$. However, its impact is negligible at moderate to high $R_t$ which justifies the tightness of the simplified bounds. For the lower bound, the same argument can be made.}

  \begin{figure}
        \centering
        \includegraphics[width=\linewidth,height=7.1cm,trim={0cm 0cm 1cm 0.7cm},clip]{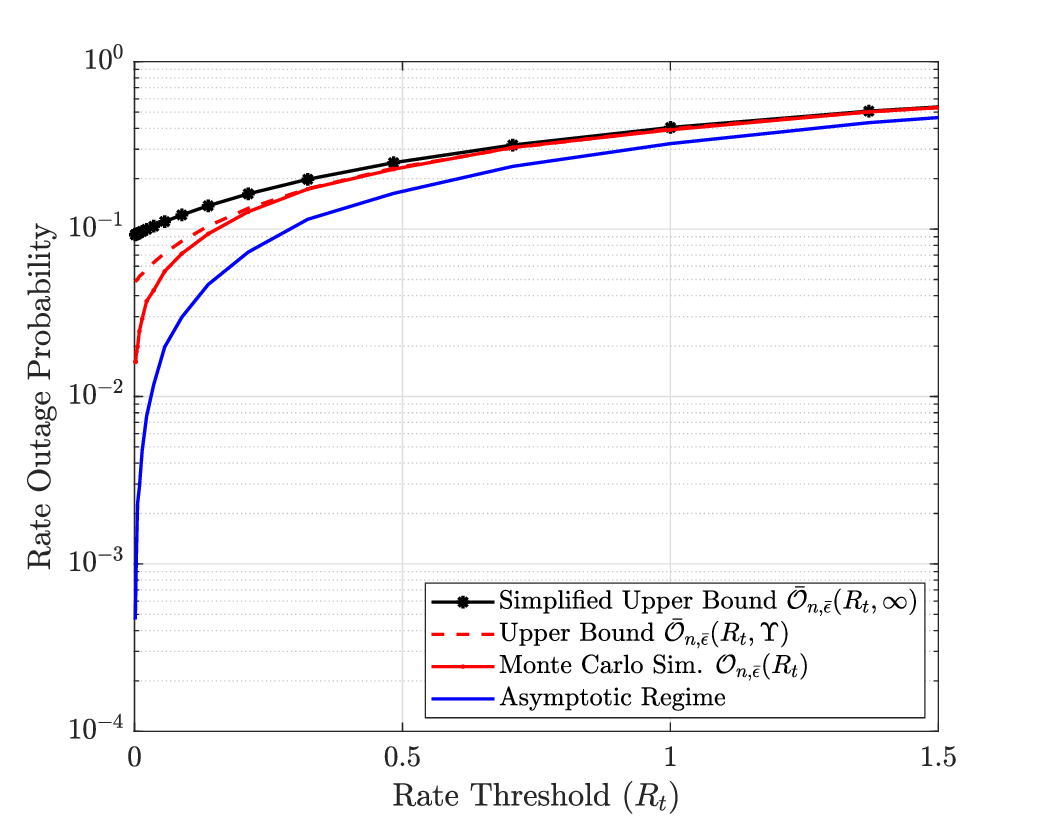}
        \caption{Outage Probability versus $R_t$ at $n=128$ and $\epsilon=10^{-2}$}
        \label{fig:Simplified_Bound}
    \end{figure}

For the meta-distribution of the coding rate simulation, \hs{\eqref{eq:meta_beta} is compared to a simulation with fixed BSs and UEs realizations}, and the fading realization changes at each time slot and each user becomes active with probability $\delta$. The SINR is evaluated for the active users and the simulation runs until $20,000$ samples of SINR are collected for each UE, which are used to evaluate per UE coding rate. The meta-distribution of the coding rate is then evaluated across all UEs. In Fig.~\ref{fig:MetaDistribution}, the coding rate meta distribution is plotted versus $p_{th}$ for rate thresholds $R_t\in\{0.5,1.2,1.5\}$ bits/transmission. It is shown that the approximation provided in \eqref{eq:meta_beta} tightly characterizes the exact coding rate meta distribution for a broad range of $R_t$.\footnote{\ac{It is worth noting that the exact meta-distribution and its lower bound, using the probability of success in \eqref{eq:Ps_exact} and \eqref{eq:lower_bound}, are numerically evaluated using Monte Carlo simulations.}} However, for small $R_t$, the approximation provides a lower bound to the coding rate meta-distribution. \ac{It is worth noting that, for all $R_t$, the approximated meta-distribution provides an achievable performance, meaning it is always guaranteed since it is always below the exact meta distribution.} Also, it is shown that the coding rate meta distribution in the AR significantly overestimates the number of users achieving the rate threshold $R_t$, by \hs{$20\%$ to $40\%$}. Hence, our analysis is more precise in estimating the users' performance.

\begin{figure}
    \centering
    \includegraphics[width=\linewidth, height=7.1cm,trim={1cm 0cm 1cm 0.7cm},clip]{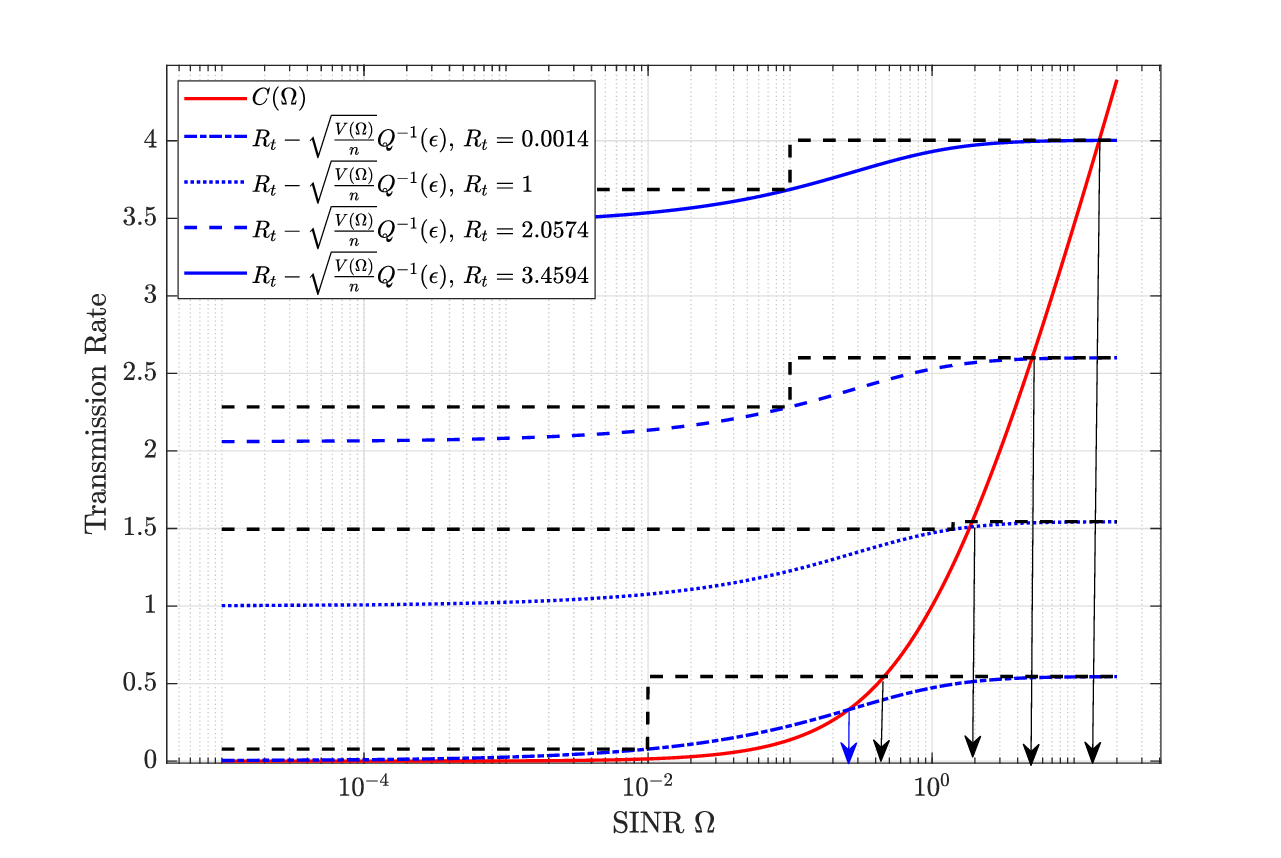}
    \caption{Outage Probability Analysis }
    \label{fig:enter-label}
\end{figure}

\begin{figure}
    \centering
    \includegraphics[width=\linewidth,height=7.1cm,trim={1cm 0cm 1cm 0.7cm},clip]{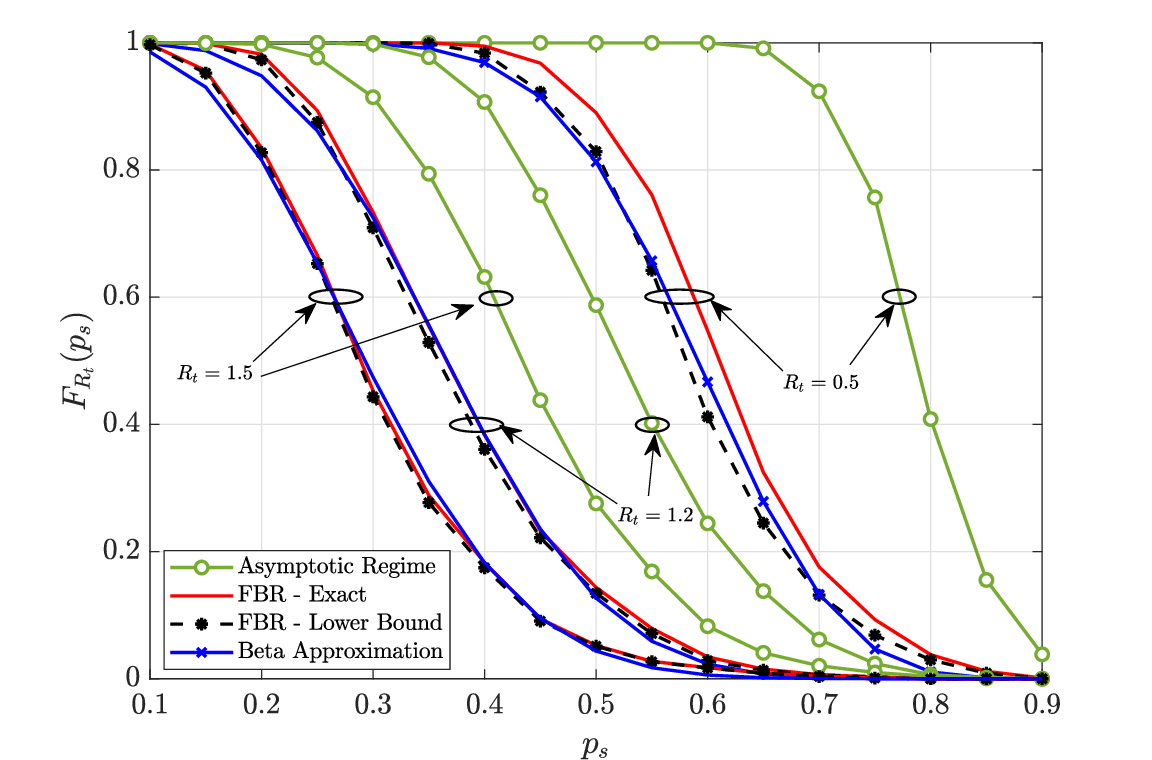}
    \caption{Rate meta distribution versus $p_{th}$ for $n=128$, $\lambda=2\ \text{BS}/\text{Km}^2$, $\bar{\epsilon}=10^{-5}$, $\alpha=1$, and $\delta=0.7$}
    \label{fig:MetaDistribution}
\end{figure}

\section{Conclusion}\label{sec:Conclusion}
We studied the performance of large-scale UL networks in the FBR in terms of the rate outage probability and the coding rate meta-distribution where bounds and approximations for {both of them are proposed}. Numerical results were provided {to highlight} the tightness of the rate outage probability upper and lower bound. Moreover, we provided a qualitative analysis for the performance of the bounds. Also, the coding rate meta-distribution approximation tightly characterizes the performance in the FBR for a large-scale network. It is also shown that the performance in the AR underestimates the rate outage probability and overestimates the coding rate meta distribution, compared to the realistic performance of the network in the FBR, and the gap is quantified. However, as the blocklength $n$ increases, the performance approaches the performance of the AR. In conclusion, our results can accurately depict the performance of large-scale networks in the FBR.

\appendices
\section{Laplace Transform of the Aggregate Interference Power}
By definition, the Laplace transform of $\mathcal{I}$ is  given by
\begin{align}
    \mathcal{L}_\mathcal{I}(s)&=\mathbb{E}\left( \exp\left\{ -s\sum_{r_i \in \Phi\setminus \{r_0\}} x_i {P_i} |h_i|^2 r_i^{-\eta}  \right\} \right).\nonumber
\end{align}
\nn{The interference coming from the interfering UE point process $\{\Phi\setminus u_0\}$ can be approximated with the interference seen from a \nn{non-homogeneous} PPP $\Tilde{\Phi}=\nn{\{}z_i\ ; i \in \mathbb{N}\}$ with intensity function $\lambda_{\text{\tiny{UE}}}(z)=\left(1-e^{-\pi\lambda ||z||^2}\right)\lambda$, where \nn{$z_i$ is the vector coordinates of the UEs} and} \nn{$\lambda_{\text{\tiny{UE}}}(z)$ models the intensity of UEs as a function of UE's location $z$ \cite{Meta}.} To simplify the analysis, \nn{we define $w=\frac{||z||^{\eta}}{P}$. Then, for an arbitrary realization, $w_i=\frac{||z_i||^{\eta}}{P_i}$ where $w_i\in \Tilde{\Phi}_o$ and follows a PPP} with the following intensity function \cite{Meta}
\begin{align}\label{eq:lambda}
    \lambda(w)=\frac{2(\pi\lambda)^{1-\alpha} \rho_o^\frac{2}{\eta}}{\eta w^{1-\frac{2}{\eta}}} \gamma\left(1+\alpha,\pi\lambda(\rho_o w)^\frac{2}{\eta(1-\alpha)}  \right),
\end{align}
which is obtained by the Displacement and Mapping Theoreoms for the PPP \cite[Theorems 2.33 \& 2.34]{reducedPalm_book}. Hence, the Laplace transform of interference is given by
\begin{align}\label{eq:LT_laststep}
        \mathcal{L}_\mathcal{I}(s)&=\mathbb{E}\left( \exp\left\{ \sum_{w_i \in \Tilde{\Phi} \setminus \{r_0\}} \frac{-s x_i  |h_i|^2}{w_i}  \right\} \right)\nonumber\\
        &=\exp\left\{- \int_{\frac{r_0^{\eta(1-\alpha)}}{\rho_o}}^\infty \lambda(w)\left( 1-\mathbb{E}\left( e^\frac{-s x |h|^2}{w}\right) \right) dw \right\}\nonumber\\
        &=\exp\left\{- \int_{\frac{r_0^{\eta(1-\alpha)}}{\rho_o}}^\infty \lambda(w)\left( 1-\frac{\delta w}{w+{s}} -(1-\delta) \right) dw \right\}\nonumber\\
        &=\exp\left\{- \int_{\frac{r_0^{\eta(1-\alpha)}}{\rho_o}}^\infty \lambda(w)\left( \frac{\delta s}{w+{s}}  \right) dw \right\}
\end{align}
By substituting \eqref{eq:lambda} in \eqref{eq:LT_laststep}, we get \eqref{eq:LT}.
\section{Rate Outage Probability Upper Bound}\label{Appendix:UpperBound}
The rate outage probability is given by
\begin{align}
    \mathcal{O}_{n,\bar{\epsilon}}(R_t)&=\mathbb{P}\left(\log_2(1+\Omega)-\frac{\sqrt{V(\Omega)}Q^{-1}(\bar{\epsilon})}{\sqrt{n}}+b_n<R_t\right)\nonumber\nonumber
    \end{align}
Define $a_{n,\bar{\epsilon}}(x)= \sqrt{ \frac{V(x)}{ n} } Q^{-1}(\bar{\epsilon})$ Noting that $a_{n,\bar{\epsilon}}(x)$ is an increasing function of $x$, then for any $\Upsilon>0$ we have    
    \begin{align}
    \mathcal{O}_{n,\bar{\epsilon}}(R_t)&=\int_0^{\Upsilon} \mathbb{P}(\Omega)\mathbb{1}(R_{n,\bar{\epsilon}}(\Omega)<R_t) d\Omega\nonumber\\
    &\hspace{1.6cm}+\int_{\Upsilon}^\infty \mathbb{P}(\Omega)\mathbb{1}(R_{n,\bar{\epsilon}}(\Omega)<R_t) d\Omega\nonumber\\
    &\hspace{-1.5cm}\leq \int_0^{\Upsilon} \mathbb{P}(\Omega)\mathbb{1}(\log_2(1+\Omega)-a_{n,\bar{\epsilon}}(\Upsilon)+b_n<R_t) d\Omega\nonumber\\
    &\hspace{-1.5cm}+\int_{\Upsilon}^\infty \mathbb{P}(\Omega)\mathbb{1}(\log_2(1+\Omega)-a_{n,\bar{\epsilon}}(\infty)+b_n<R_t) d\Omega\nonumber
\end{align}

Hence, the rate outage probability is upper bounded by
\begin{align}
    \bar{\mathcal{O}}_{n,\bar{\epsilon}}(R_t,\Upsilon)&=\int_0^{\beta_1} \mathbb{P}(\Omega) d\Omega+\int_{\Upsilon}^{\beta_2} \mathbb{P}(\Omega) d\Omega\nonumber\\
    &=\mathbb{F}_{\Omega}(\beta_1)+\mathbb{F}_{\Omega}(\beta_2)-\ac{\mathbb{F}_{\Omega}(\Upsilon)}\nonumber
\end{align}
for any $\Upsilon>0$, where $\beta_1=\min(\Upsilon,2^{R_t+a_{n,\bar{\epsilon}}(\Upsilon)-b_n}-1)$ and $\beta_2=\max(\Upsilon,2^{R_t+a_{n,\bar{\epsilon}}(\infty)-b_n}-1)$. This concludes the derivation and leads to the result in \eqref{eq:Outage_UB}.
\section{Rate Outage Probability Lower Bound}\label{Appendix:LowerBound}
The rate outage probability is given by $\mathcal{O}_{n,\bar{\epsilon}}(R_t)=\mathbb{P}\left(R_{n,\bar{\epsilon}}(\Omega)<R_t\right)$. Then, since $a_{n,\bar{\epsilon}}(x)$ is increasing in $x$, for any $\Lambda>0$ we have
\begin{align}
 & {\mathcal{O}}_{n,\bar{\epsilon}}(R_t)=\int_0^{\Lambda} \mathbb{P}_\Omega(\Omega) \mathbb{1}(R_{n,\bar{\epsilon}}(\Omega)<R_t) d\Omega\nonumber\\
  &\hspace{3cm}+ \int_{\Lambda}^\infty \mathbb{P}_\Omega(\Omega) \mathbb{1}(R_{n,\bar{\epsilon}}(\Omega)<R_t) d\Omega\nonumber\\
  &\geq \int_0^{\Lambda} \mathbb{P}_\Omega(\Omega) \mathbb{1}(\log_2(1+\Omega)-a_{n,\bar{\bar{\epsilon}}}(0)+b_n<R_t) d\Omega\nonumber\\
  &+ \int_{\Lambda}^\infty \mathbb{P}_\Omega(\Omega) \mathbb{1}(\log_2(1+\Omega)-a_{n,\bar{\bar{\epsilon}}}(\Lambda)+b_n<R_t) d\Omega.\nonumber
  \end{align}
  Define $\mu_1=\min(\Lambda,2^{R_t-b_n}-1)$ and $\mu_2=\max(\Lambda,2^{R_t+a_{n,\bar{\epsilon}}(\Lambda)-b_n}-1)$, and since $a_{n,\bar{\bar{\epsilon}}}(0)=0$, the rate outage probability is lower bounded by
  \begin{align}
 \hspace{-5cm}\underline{\mathcal{O}}_{n,\bar{\epsilon}}(R_t) &=\int_0^{\mu_1} \mathbb{P}_\Omega(\Omega) d\Omega +\int_{\Lambda}^{\mu_2} \mathbb{P}_\Omega(\Omega) d\Omega \ \ \ \nonumber\\
  &=\mathbb{F}_\Omega(\mu_1)-\mathbb{F}_\Omega(\Lambda)+\mathbb{F}_\Omega(\mu_2),\nonumber
\end{align}
for any $\Lambda>0$.This concludes the derivation and leads to \eqref{eq:Outage_LB}.
%This concludes the derivation and leads to the result in \eqref{eq:Outage_LB}.

\bibliographystyle{IEEEtran}
\bibliography{References}

\end{document}